\documentclass[twocolumn]{autart}

\usepackage{amsmath}
\usepackage{algorithm,algpseudocode}
\usepackage{amssymb}
\usepackage{color}
\usepackage{enumerate}
\usepackage{wrapfig}
\usepackage{graphicx}          
\usepackage[dvips]{epsfig}    
\usepackage[round, sort]{natbib}
\usepackage{hyperref}
\usepackage{comment}
\usepackage{mathrsfs}
\usepackage{xpatch}

\makeatletter
\xpatchcmd{\no@harm}
  {\def\protect{\noexpand\protect\noexpand}}
  {}
  {\typeout{autart protect patch applied}}
  {\typeout{autart protect patch failed}}
\makeatother

\usepackage{url}

\newtheorem{lemma}{Lemma}[section]
\newtheorem{theorem}{Theorem}[section]

\newtheorem{remark}{Remark}[section]
\newtheorem{definition}{Definition}[section]

\newtheorem{assumption}{Assumption}[section]

\newtheorem{problem}{Problem}

\begin{document}
\begin{frontmatter}
			
	\title{Data-Driven Spiking Control for Distributed $\epsilon$-Nash Equilibrium Seeking}

	\thanks[footnoteinfo]{The work was supported in part by the National Natural Science Foundation of China under Grants U23B2059 and 62688201. (\emph{Corresponding author: Gang Wang.})}

		\author[Bit]{Linqi Wang}\ead{wanglinqi@bit.edu.cn},
        \author[Bit]{Wei Xiao}\ead{xiaowei@bit.edu.cn},
        \author[Bit]{Yuzhou Wei}\ead{weiyuzhou@bit.edu.cn}, 
        \author[Bit]{Bin Xin}\ead{brucebin@bit.edu.cn},  
        \author[Thu]{Keyou You}\ead{youky@tsinghua.edu.cn},  
        		\author[Bit]{Gang Wang}\ead{gangwang@bit.edu.cn}
		
		\address[Bit]{National Key Lab of Autonomous Intelligent Unmanned Systems, Beijing Institute of Technology, Beijing 100081, China}  
       \address[Thu]{Department of Automation and BNRist, Tsinghua University, Beijing 100084, China}

		
		\maketitle


\begin{abstract}
This paper studies how a feedback law synthesized directly from data can be realized by spiking control while retaining a game-theoretic performance guarantee. We consider distributed $\epsilon$-Nash equilibrium (NE) seeking in network games played by linear dynamical agents with unknown models and exogenous disturbances. The pseudo-gradient of the game is treated as a regulated error, and local internal models account for signals generated by known exosystems. Robust linear matrix inequalities are then used to compute stabilizing analogue feedback gains directly from noisy local input-state data, without identifying the agent dynamics. To implement these gains using only fixed-weight spikes, we develop two spiking realizations. The first realization uses non-interacting leaky integrate-and-fire units, while the second permits reset coupling among the neuronal units. In both cases, a continuous auxiliary coordinate exposes the impulsive closed loop as the stable analogue system driven by a bounded implementation error. This representation yields forward completeness, Zeno-freeness, and an ultimate bound on the pseudo-gradient, subject to explicit event-processing conditions for the connected architecture. The bound implies that, after a finite transient, the agents' outputs constitute an $\epsilon$-NE for every $\epsilon$ above a finite threshold. A spacecraft formation reconfiguration example illustrates the data-driven synthesis, the two spiking realizations, and their practical equilibrium behavior.
\end{abstract}
		\end{frontmatter}


    \allowdisplaybreaks

\section{Introduction}
Network games provide a natural model for decision making in
interconnected engineering systems. Examples include power and
transportation networks, distributed sensing platforms, and robotic
teams. More recently, this perspective has been extended to embodied
intelligent games involving autonomous interactions among
heterogeneous agents \citep{yi2025embodied}. In a network game, each
agent optimizes a local objective that depends on both its own
decision and the decisions communicated by its neighbors
\cite{bacsar1998dynamic,stankovic2011distributed}. The desired operating point is typically a Nash equilibrium (NE), at which no agent can improve its objective by unilaterally deviating
\cite{ABOUHEAF20143038,liu2025online,Decentralized2024Meng}.
Computing such a profile in a distributed manner requires one to respect both the noncooperative objectives of the agents and the information constraints imposed by the communication graph. Gradient-play, consensus-based, learning-based, and passivity-based methods provide several ways to address this problem \cite{SALEHISADAGHIANI201817,de2019distributed,zhu2020distributed,gadjov2018passivity,liu2024distributed,liu2025distributed,li2024master,nortmann2024nash,lian2025distributed}.

When each decision is the output of a dynamical system, equilibrium
seeking becomes a control problem. The agents must remain stable while
their outputs are steered toward a game-theoretic solution
\cite{9863764,tang2022nash,WangRobus2022}.
Related feedback Nash
designs have been developed for formation control and collision
avoidance of multiple unmanned aerial vehicles in unknown
environments \citep{xue2025formation}, illustrating the broader role
of game-theoretic feedback in coordinating dynamical multi-agent
systems.
For the quadratic network games studied
here, the pseudo-gradient is affine, and its unique zero coincides
with the NE. This observation allows NE seeking to be formulated as an output regulation problem, with the pseudo-gradient selected as the regulated error and the internal model principle employed to reject constant or harmonic disturbances generated by known exosystems
\cite{huang2004nonlinear,guo2021linear,bin2020approximate,liu2025data}.
An exact realization would drive this error to zero. With persistent
uncertainty or implementation mismatch, however, practical
regulation is the appropriate objective. A bounded pseudo-gradient further yields an explicit $\epsilon$-NE guarantee, ensuring that no agent can reduce its cost by more than $\epsilon$ through a unilateral deviation
\cite{MylvaganamNash2015,caines2016epsilon}.

The first obstacle is that regulation-based designs normally require a model of every agent. Such models may be costly to derive and may not represent the operating system accurately. Direct data-driven
control instead uses measured trajectories as the description from which a controller is synthesized and certified \citep{willems2005note}. In networked settings, the resulting literature spans communication delays, aperiodic transmission, security, and distributed configurations \citep{wang2026data}. Specific developments include stabilization, predictive control, adaptive control, and output regulation of unknown linear systems \citep{berberich2020data,liu2022data,zhao2025data,li2026data}. It has also begun to influence game-theoretic control. Trajectory-based and iterative schemes have been proposed for linear quadratic dynamic games \citep{nortmann2024nash}, and distributed NE seeking has recently been formulated as a data-driven output regulation problem \citep{wang2026robust}. These methods, however, still implement the resulting feedback as a continuously valued input and therefore require online modulation of the control amplitude.

The second obstacle is therefore to realize the data-designed feedback without applying its continuously modulated amplitude directly to the plant. Motivated by spiking control \cite{sepulchre2022spiking,PETRI2026101759}, this paper considers an implementation in which the feedback command drives a collection of leaky integrate-and-fire (LIF) units and is encoded by their firing times and channel selections. Each channel is associated with a prescribed impulse vector, so that a threshold crossing generates a fixed-weight spike along the corresponding direction. Here, spiking control describes the neuronal implementation, in which LIF dynamics serve as causal event generators driven by the feedback signal \cite{gerstner2002spiking}, whereas impulsive control provides the mathematical framework for the resulting state jumps at isolated event times \cite{meng2012optimal,yang2002impulsive,li2025nonlinear}. Existing spiking and neuronal control results, however, mainly start from a known plant model or a predesigned analogue controller \cite{petri2025spiking,eilers2025stability,slijkhuis2023closed}. They therefore do not establish whether the stability and game-theoretic guarantees of the data-designed analogue closed loop persist under spiking implementation. In particular, spike-induced plant jumps are coupled with neuronal resets, while reset interactions may trigger firing cascades. A framework connecting robust synthesis from noisy data, spiking implementation, impulsive closed-loop analysis, and \(\epsilon\)-NE performance is therefore still lacking.

This paper develops such a framework. We first formulate distributed
$\epsilon$-NE seeking as practical regulation of the pseudo-gradient.
Each agent augments its unknown dynamics with a local internal model
and collects an offline trajectory in isolation. A robust data-based
linear matrix inequality then yields a stabilizing analogue feedback
gain without identifying the agent dynamics.  Online, this gain
determines the input functions of the neuronal units, whose spikes form the fixed-weight control input applied to the agent. The key analytical device is an auxiliary
coordinate that combines the agent and neuronal states so that spike
jumps cancel exactly. In this coordinate, the impulsive system becomes
the stable analogue closed loop perturbed by a bounded neuronal
state. This interpretation makes the relation between data-driven synthesis, spiking implementation, and practical game performance.

The main contributions are as follows.
\begin{enumerate}
\item We develop a direct data-driven spiking control framework for distributed \(\epsilon\)-NE seeking. Stabilizing analogue gains are synthesized robustly from noisy local trajectories without model identification, while the resulting commands are encoded by the firing times and channel selections of LIF-generated spikes rather than by continuously varying the control amplitude.
\item For non-interacting LIF units, we derive an exact rate-equivalent encoding of the data-designed feedback and construct a continuous auxiliary coordinate that cancels spike-induced jumps. We establish forward completeness and positive finite-horizon inter-event-time bounds, and derive an ultimate pseudo-gradient bound yielding an eventual $\epsilon$-NE guarantee.
\item We further develop a reset-coupled spiking realization that implements the same data-designed gain using a smaller set of generally nonorthogonal impulse directions. Under explicit finite-cascade and nonaccumulation conditions, we establish boundedness of the coupled neuronal states and derive the corresponding closed-loop and $\epsilon$-NE guarantees.
\end{enumerate}

The remainder of this paper is organized as follows.
Section~\ref{Problem Formulation} formulates the game, develops its
regulation representation, describes the offline data, and introduces the spiking realization. Sections~\ref{Data-Driven Impulsive Control}
and~\ref{Extension to Connected Neuronal Units} analyze the
non-interacting and connected neuronal architectures, respectively.
Section~\ref{sec:simulation} presents the spacecraft example, and
Section~\ref{conclusion} concludes the paper.

\emph{Notation.} Let $\mathbb{R}^+$ denote the set of nonnegative numbers.
The symbol $\|\cdot\|$ denotes the Euclidean norm for vectors and the induced $2$-norm for matrices.
For a matrix $A$, $A^\top$, $\operatorname{spec}(A)$, and
$\operatorname{im}(A)$ denote its transpose, spectrum, and
column space, respectively, and $I$ denotes an identity
matrix of compatible dimension.
The operators $\operatorname{col}(\cdot)$, $\operatorname{diag}(\cdot)$ and $\operatorname{blkdiag}(\cdot)$ denote column stacking, a diagonal
matrix and a block-diagonal matrix, respectively.
For symmetric matrix $A$, the relations $A\succ 0$ ($A\succeq 0$) and $A\prec 0$ ($A\preceq 0$) denote positive (semi)definiteness and negative (semi)definiteness, respectively, while $\lambda_{\min}(A)$ and $\lambda_{\max}(A)$ denote its smallest and largest eigenvalues.
For $A\succ 0$, let $\kappa(A)=\lambda_{\max}(A)/\lambda_{\min}(A)$.
A square matrix is Hurwitz if all its eigenvalues have negative real parts.
For $x\in\mathbb{R}$, define $[x]_+:=\max\{x,0\}$, with componentwise extension to vectors.
At event time $t$, $t^-$ and $t^+$ denote the left and right limits.
All impulsive trajectories are right-continuous, and a solution is called Zeno if it has infinitely many events in finite time.
The symbol $\mathbf{1}$ denotes an all-one vector of compatible dimension. For $M\in\mathbb R^{m\times L}$,
$\operatorname{cone}(M):=\{M\rho:\rho\in\mathbb R_+^L\}$
denotes the conic hull of the columns of $M$. In symmetric block matrices, $*$ denotes the symmetric
counterpart of the corresponding off-diagonal block.

\section{Problem Formulation and Control Architecture}\label{Problem Formulation}
We develop the proposed architecture in four steps. We first define
the network game and relate its pseudo-gradient to the
$\epsilon$-NE objective. We then lift this objective to an output
regulation problem, construct a data-consistent model set from local
experiments, and finally explain how the data-driven feedback will be implemented through neuronal spike trains.
\subsection{Network Game and Practical Equilibrium}
\label{subsec:network_game}
Consider \(N\) agents indexed by
\(\mathcal N:=\{1,\ldots,N\}\). Agent \(i\in\mathcal N\) is described by
\begin{subequations}    \label{eq:agent_dynamics}
\begin{align}
    \dot x_i&=A_i x_i+B_i u_i+W_i\omega_i,\\
    y_i&=C_i x_i,
\end{align}
\end{subequations}
where \(x_i\in\mathbb{R}^{n_i}\), \(u_i\in\mathbb{R}^{m_i}\), and
\(y_i\in\mathbb{R}^{p_i}\) denote the state, control input, and
decision output, respectively. The constant matrices \(A_i\), \(B_i\),
\(C_i\), and \(W_i\) are unknown. The disturbance
\(\omega_i\in\mathbb{R}^{q_i}\) is generated by the known exosystem
\begin{equation}
    \dot \omega_i=E_i\omega_i.
    \label{eq:exosystem}
\end{equation}
\begin{assumption}
\label{ass:exosystem}
For each \(i\in\mathcal N\), \(E_i\) has no eigenvalues with
negative real parts.
\end{assumption}

Information is exchanged over a directed graph
$\mathcal{G}_c=(\mathcal{N},\mathcal{E})$. An edge
$(j,i)\in\mathcal{E}$ means that agent $i$ receives the decision
output of agent $j$, and
$\mathcal{N}_i:=\{j\in\mathcal{N}:(j,i)\in\mathcal{E}\}$ is the
corresponding in-neighbor set.
\begin{assumption}\label{ass:graph}
The graph $\mathcal{G}_c$ is weakly connected and acyclic.
\end{assumption}
Weak connectivity is commonly assumed for network games over directed
graphs \cite{guo2021linear}.  Weak connectivity ensures that
the underlying undirected graph is connected, whereas acyclicity
guarantees the existence of a topological ordering of the agents.
The role of this ordering in the equilibrium analysis will be made
explicit below.

Given the decisions of its in-neighbors, agent \(i\) determines its decision \(y_i\) by minimizing the cost function
\begin{equation}\label{eq:cost}
 J_i(y_i,y_{\mathcal{N}_i})
 =y_i^{\top}R_{ii}y_i+Q_{ii}y_i
 +\sum_{j\in\mathcal{N}_i}y_i^{\top}R_{ij}y_j,
\end{equation}
where
$y_{\mathcal{N}_i}:=\operatorname{col}(y_j)_{j\in\mathcal{N}_i}$,
$R_{ii}=R_{ii}^{\top}\succ0$,
$Q_{ii}\in\mathbb{R}^{1\times p_i}$, and
$R_{ij}\in\mathbb{R}^{p_i\times p_j}$. We set \(R_{ij}=0\) whenever
\(j\notin\mathcal N_i\), \(j\ne i\). The cost matrices are known to
agent \(i\). We denote the resulting network game by
$\mathcal{G}:=\bigl(\mathcal{N},\mathbb{R}^{p_i},J_i,\mathcal{G}_c\bigr)$.

\begin{definition}\cite[Definition~4.1]{bacsar1998dynamic}\label{def:NE}
A profile $y^\star={\rm col}(y_1^\star,\ldots,y_N^\star)$ is an NE of $\mathcal{G}$ if, for every $i\in\mathcal{N}$,
\begin{align*}
J_i(y_i^\star,y_{\mathcal{N}_i}^\star)
 \leq J_i(y_i,y_{\mathcal{N}_i}^\star),\qquad
 \forall y_i\in\mathbb{R}^{p_{i}}.
 \end{align*}
\end{definition}
Let $y:={\rm col}(y_1,\ldots,y_N)$. The pseudo-gradient is
\begin{align}\label{eq:pseudogradient}
 F(y)&:={\rm col}\bigl(\nabla_1J_1(y_1,y_{\mathcal{N}_1}),\ldots,\nabla_NJ_N(y_N,y_{\mathcal{N}_N})\bigr)\nonumber\\
 &=\widetilde R y+\widetilde Q,
\end{align}
where $ \widetilde Q:=\operatorname{col}(Q_{11}^{\top},\ldots,Q_{NN}^{\top})$,  and \begin{equation*}
 [\widetilde R]_{ij}=
 \begin{cases}
 R_{ii}+R_{ii}^{\top},&j=i,\\
 R_{ij},&j\in\mathcal{N}_i,\\
 0,&\text{otherwise}.
 \end{cases}
\end{equation*}
Under a topological ordering of the acyclic graph, \(\widetilde R\) is
block lower triangular, with diagonal blocks
$R_{ii}+R_{ii}^{\top}=2R_{ii}\succ0$. It is therefore nonsingular.
Because \(J_i(\cdot,y_{\mathcal N_i})\) is strictly convex, the
first-order condition \(F(y)=0\) is also sufficient for an NE.
Consequently, the game admits the unique equilibrium
\begin{equation}
\label{eq:unique_NE}
y^\star=-\widetilde R^{-1}\widetilde Q.
\end{equation}
The spiking realization developed below need not reproduce the analogue feedback exactly. We therefore use the following approximate
equilibrium notion to express the residual closed-loop error.

\begin{definition}\cite[p. 13]{etessami2010complexity}
\label{def:epsilon_NE}
For \( \epsilon\ge0\), a profile
\(y^ \epsilon=\operatorname{col}(y_1^ \epsilon,\ldots,y_N^ \epsilon)\)
is an \( \epsilon\)-NE of \(\mathcal G\) if, for every $i\in\mathcal{N}$,
\begin{equation*}
J_i(y_i^ \epsilon,y_{\mathcal N_i}^ \epsilon)
\le J_i(y_i,y_{\mathcal N_i}^ \epsilon)+ \epsilon,
\qquad \forall y_i\in\mathbb R^{p_i}.
\end{equation*}
\end{definition}
The next subsection converts this cost-based definition into a
regulation objective that is directly amenable to controller design.

\subsection{Regulation Lifting and Analogue Reference Controller}
\label{subsec:practical_output_regulation}
Since the action sets are unconstrained and
$J_i(\cdot,y_{\mathcal N_i})$ is strongly convex for every
$i\in\mathcal N$, a NE $y^\star$ is characterized by
$F(y^\star)=0$. This characterization allows the NE
seeking problem to be formulated as a regulation problem. Accordingly,
we select the pseudo-gradient residual as the regulated output
\begin{equation}
    e:=F(y)=\widetilde R y+\widetilde Q, ~~e_i=\nabla_iJ_i(y_i,y_{\mathcal N_i}).
    \label{eq:regulated_error}
\end{equation}

We next quantify how the regulation accuracy determines the NE
accuracy. For each \(i\in\mathcal N\), define
\[
S_i:=R_{ii}+R_{ii}^{\top},
\qquad
\alpha_i:=\lambda_{\min}(S_i)>0.
\]
For fixed \(y_{\mathcal N_i}\), completing the square yields
\begin{equation}\label{eq:individual_regret_bound}
J_i(y_i,y_{\mathcal N_i})
-\inf_{\bar y_i\in\mathbb R^{p_i}}
J_i(\bar y_i,y_{\mathcal N_i})
=
\frac{1}{2}e_i^{\top}S_i^{-1}e_i
\le
\frac{\|e_i\|^2}{2\alpha_i}.
\end{equation}

Let
$\underline{\alpha}
:=
\min_{i\in\mathcal N}\alpha_i$.
It follows from \eqref{eq:individual_regret_bound} that
\begin{equation}\label{eq:aggregate_regret_bound}
\max_{i\in\mathcal N}
\left\{
J_i(y_i,y_{\mathcal N_i})
-\inf_{\bar y_i\in\mathbb R^{p_i}}
J_i(\bar y_i,y_{\mathcal N_i})
\right\}
\le
\frac{\|F(y)\|^2}{2\underline{\alpha}}.
\end{equation}
Therefore, the practical regulation condition
\begin{equation}\label{eq:practical_regulation}
\limsup_{t\to\infty}\|e(t)\|\le\delta
\end{equation}
implies
\[
\limsup_{t\to\infty}
\max_{i\in\mathcal N}
\{
J_i(y_i(t),y_{\mathcal N_i}(t))
\!-\!\inf_{\bar y_i\in\mathbb R^{p_i}}
J_i(\bar y_i,y_{\mathcal N_i}(t))
\}
\!\le\!
\frac{\delta^2}{2\underline{\alpha}}.
\]
Thus, a pseudo-gradient residual bounded by \(\delta\) yields the
corresponding asymptotic NE-accuracy bound.

The preceding result establishes a direct link between regulation
performance and NE accuracy. To obtain a standard local
output-regulation model, the exogenous signal $\omega_i$ and the
constant term $Q_{ii}^{\top}$ are embedded into the extended
exogenous state
\begin{equation}\label{eq:extended-exosystem}
 \nu_i:={\rm col}(\omega_i,1),\qquad
 \dot\nu_i=\bar E_i\nu_i,
\end{equation}
where
$ \bar E_i:=
\begin{bmatrix}
\begin{smallmatrix}
E_i&0_{q_i\times1}\\0_{1\times q_i}&0
\end{smallmatrix}\end{bmatrix}$, 
$\nu_i(0)={\rm col}(\omega_i(0),1)$.
Define $\bar W_i:=[\,W_i\ \ 0_{n_i\times1}\,]$ and $\bar Q_{ii}:=[\,0_{p_i\times q_i}\ \ Q_{ii}^{\top}\,]$. The local dynamics and regulated error can then be written as
\begin{subequations}\label{eq:local-regulation-model}
\begin{align}
 \dot x_i&=A_ix_i+B_iu_i+\bar W_i\nu_i,\\
 \dot\nu_i&=\bar E_i\nu_i,\\
 e_i&=(R_{ii}+R_{ii}^{\top})C_ix_i+
 \sum_{j\in\mathcal N_i}R_{ij}C_jx_j+\bar Q_{ii}\nu_i,\label{eq:local-regulation-modelc}
\end{align}
\end{subequations}
We use the following standard stabilization and output regulation
conditions.
\begin{assumption}
\label{ass:stabilizability_detectability}
For each \(i\in\mathcal N\), the pair \((A_i,B_i)\) is stabilizable.
\end{assumption}
\begin{assumption} \label{ass:regulation_solvability} For each \(i\in\mathcal N\) and every \(\lambda\in\operatorname{spec}(E_i)\cup\{0\}\),
$\operatorname{rank}\begin{bmatrix}
\begin{smallmatrix} A_i-\lambda I_{n_i} & B_i\\ C_i & 0_{p_i\times m_i}
 \end{smallmatrix}\end{bmatrix}=n_i+p_i$.
\end{assumption}
\begin{remark}
Assumption \ref{ass:regulation_solvability} is the standard transmission-zero condition for linear output regulation
\cite[Theorem~1.9, Remark~1.11]{huang2004nonlinear}. It excludes invariant zeros of $(A_i,B_i,C_i)$ at all eigenvalues of the extended exosystem, including the zero eigenvalue introduced by the constant offset, and thereby guarantees solvability of the local regulator equations. Since $R_{ii}+R_{ii}^{\top}\succ0$, replacing $C_i$ with $(R_{ii}+R_{ii}^{\top})C_i$ preserves the rank condition; hence the same solvability condition applies to the local regulated error in \eqref{eq:local-regulation-model}.
\end{remark}

To reject signals generated by the extended exosystem, we invoke the
internal model principle \cite{huang2004nonlinear}. Agent \(i\) is
equipped with a controllable \(p_i\)-copy internal model
\((\Pi_i,\widehat{\Pi}_i)\) of \(\bar E_i\)
\begin{equation}
\dot\zeta_i=\Pi_i\zeta_i+\widehat{\Pi}_i e_i,~~~ \zeta_i\in\mathbb{R}^{h_i}. \label{eq:internal_model} \end{equation}
Define the local augmented state $\eta_i:=\operatorname{col}(x_i,\zeta_i)$.
For compactness, introduce the stacked variables
\begin{align*}
x&:=\operatorname{col}(x_1,\ldots,x_N),&
\zeta&:=\operatorname{col}(\zeta_1,\ldots,\zeta_N),\\
\eta&:=\operatorname{col}(\eta_1,\ldots,\eta_N),&
\nu&:=\operatorname{col}(\nu_1,\ldots,\nu_N),\\
u&:=\operatorname{col}(u_1,\ldots,u_N).
\end{align*}
Combining \eqref{eq:local-regulation-model} and \eqref{eq:internal_model} gives the augmented  dynamics
\begin{subequations} \label{eq:augmented_open_loop} 
\begin{align} \dot\eta &= A_\eta\eta+B_\eta u+W_\eta\nu,\label{eq:augmented_open_loop1}\\
\dot\nu &= \bar E\nu,\label{eq:augmented_open_loop2}\\ 
e &= C_\eta\eta+Q_\eta\nu\label{eq:augmented_open_loop3},
\end{align}
\end{subequations}
where $\bar E:=\operatorname{blkdiag}(\bar E_1,\ldots,\bar E_N)$,
\(A_\eta=[A_{\eta,ij}]_{i,j=1}^N\) is block lower triangular, \(C_\eta=[C_{\eta,ij}]_{i,j=1}^N\), and
\begin{align*}
A_{\eta,ii}
&:=
\begin{bmatrix}
A_i & 0\\
\hat{\Pi}_{i}(R_{ii}+R_{ii}^{\top})C_i & \Pi_{i}
\end{bmatrix},~~B_{\eta i}:=
\begin{bmatrix}
B_i\\
0
\end{bmatrix},\\
A_{\eta,ij}
&:=
\begin{cases}
\begin{bmatrix}
0 & 0\\
\hat{\Pi}_{i}R_{ij}C_j & 0
\end{bmatrix},
& j\in\mathcal N_i,\\[1mm]
0, & \text{otherwise},
\end{cases}~~
W_{\eta i}:=
\begin{bmatrix}
\bar W_i\\
\hat{\Pi}_{i}\bar Q_{ii}
\end{bmatrix},\\
C_{\eta,ij}&:=
\begin{cases}
\begin{bmatrix}
(R_{ii}+R_{ii}^{\top})C_i & 0
\end{bmatrix},
& i=j,\\[1mm]
\begin{bmatrix}
R_{ij}C_j & 0
\end{bmatrix},
& j\in\mathcal N_i,\\[1mm]
0, & \text{otherwise},
\end{cases}\\
B_\eta&:=\operatorname{blkdiag}(B_{\eta1},\ldots,B_{\eta N}),\\
W_\eta&:=\operatorname{blkdiag}(W_{\eta1},\ldots,W_{\eta N}), \\
Q_\eta&:=\operatorname{blkdiag}
(\bar Q_{11},\ldots,\bar Q_{NN}).
\end{align*}
If the model were available and continuous actuation were permitted,
a natural reference design would be the distributed state feedback
\begin{equation}\label{eq:analogue-feedback}
 u_i=K_i\eta_i=K_{xi}x_i+K_{\zeta i}\zeta_i,
~ K_i:=[\,K_{xi}\ \ K_{\zeta i}\,].
\end{equation}
We call \eqref{eq:analogue-feedback} the \emph{analogue reference
controller}. The gain \(K_i\) will be computed from data and used to
drive the neuronal encoder; the signal \(K_i\eta_i\) itself is not
applied continuously to the system.

\subsection{Local Data and the Consistent Model Set}\label{Local Data}
Since the local augmented matrices $A_{\eta,ii}$ and $B_{\eta i}$
are unknown, direct model-based synthesis is not available. Moreover,
using the full regulated error during data acquisition would make the
augmented trajectory depend on neighboring outputs. A finite local
experiment is therefore performed prior to controller deployment.

During the experiment, the full regulated error in
\eqref{eq:local-regulation-modelc} is replaced by its local component
\begin{equation}\label{eq:local_data_error}
e_i^{\mathrm{loc}}
:=
(R_{ii}+R_{ii}^{\top})y_i+Q_{ii}^{\top}.
\end{equation}
This amounts to setting
$\sum_{j\in\mathcal N_i}R_{ij}y_j$ to zero at the input of the
internal model. Since the plant dynamics contain no inter-agent
coupling, only the incoming information channel is modified, and no
physical disconnection is required.

Let $\bar\zeta_i$ denote the internal-model state used during data
acquisition, governed by
$\dot{\bar\zeta}_i:=\Pi_i\bar\zeta_i+\widehat{\Pi}_i e_i^{\mathrm{loc}}$,
$\bar\eta_i:=\operatorname{col}(x_i,\bar\zeta_i)$.
Substitution of \eqref{eq:local_data_error} into the internal-model
dynamics yields
\begin{equation}\label{eq:offline_augmented_dynamics}
\dot{\bar\eta}_i
=
A_{\eta,ii}\bar\eta_i
+
B_{\eta i}u_i
+
W_{\eta i}\nu_i.
\end{equation}
At sampling instants \(t_0,\ldots,t_T\), we record
$\{u_i(t_\ell)\}_{\ell=0}^{T-1}$ and
$\{x_i(t_\ell),y_i(t_\ell),\bar{\zeta}_i(t_\ell)\}_{\ell=0}^{T}$.
The state derivatives are approximated by
$\dot x_i(t_\ell):=\frac{x_i(t_{\ell+1})-x_i(t_\ell)}
{t_{\ell+1}-t_\ell}$ and
$\dot{\bar \zeta}_{i}(t_\ell):=\frac{\bar{\zeta}_i(t_{\ell+1})-\bar{\zeta}_i(t_\ell)}
{t_{\ell+1}-t_\ell}$. We then set
$\dot{\bar{\eta}}_i(t_\ell)
:=
\operatorname{col}
\bigl(
 \dot{x}_i(t_\ell),
 \dot{\bar{\zeta}}_i(t_\ell)
\bigr)$.
Define the data matrices
\begin{subequations}
\label{eq:data_matrices}
\begin{align}
U_i^-&:=[u_i(t_0)\;u_i(t_1)\;\cdots\;u_i(t_{T-1})],\\
X_i^-&:=[\bar{\eta}_i(t_0)\;\bar{\eta}_i(t_1)\;\cdots\;\bar{\eta}_i(t_{T-1})],\\
X_i^+&:=[\dot{\bar{\eta}}_i(t_0)\;\dot{\bar{\eta}}_i(t_1)\;\cdots\;
\dot{\bar{\eta}}_i(t_{T-1})].
\end{align}
\end{subequations}
The collected data satisfy
\begin{equation}\label{eq:data_equation}
    X_i^+=A_{\eta,ii}X_i^-+B_{\eta i}U_i^-+W_{\eta i}V_i+D_i,
\end{equation}
where $V_i:=[\nu_i(t_0)\;\nu_i(t_1)\;\cdots\;\nu_i(t_{T-1})]$.
The matrix
$D_i:=[d_i(t_0)\;d_i(t_1)\;\cdots\;d_i(t_{T-1})]$ collects derivative
approximation and measurement errors. The unmeasured exogenous
contribution and these errors are treated jointly through the
following assumptions.
\begin{assumption}\label{ass:rank_data}
For each \(i\in\mathcal N\), the matrix
$Z_i:=\operatorname{col}(X_i^-,U_i^-)$
has full row rank.
\end{assumption}
\begin{assumption}\label{ass:noise_bound}
For each \(i\in\mathcal N\), the aggregate uncertainty
\(\hat{D}_{i}:=W_{\eta i}V_i+D_i\) is bounded: there exists a known
matrix \(\Delta_i\) such that \(\hat{D}_{i}\in\mathcal{D}_i\), where
\begin{equation}\label{eq:noise_bound}
\mathcal{D}_{i}:=\left\{\bar{D}_{i}\in
\mathbb{R}^{(n_i+h_i)\times T}:
\bar{D}_{i}\bar{D}_{i}^\top
\preceq\Delta_{i}\Delta_{i}^\top\right\}.
\end{equation}
\end{assumption}
Let
\(\Gamma_i^\top:=[\bar A_{\eta,ii}\;\;\bar B_{\eta i}]\) denote a
candidate pair of local augmented matrices. The data-consistent model
set is
\begin{align}
\mathcal C_i
={}&
\Bigl\{
\Gamma_i^\top
\;\Big|\;
(X_i^+-\Gamma_i^\top Z_i)
(X_i^+-\Gamma_i^\top Z_i)^\top
\preceq
\Delta_i\Delta_i^\top
\Bigr\}
\nonumber\\
={}&
\Bigl\{
\Gamma_i^\top
\;\Big|\;
\Omega_i
+\Psi_i^\top\Gamma_i
+\Gamma_i^\top\Psi_i
+\Gamma_i^\top\Phi_i\Gamma_i
\preceq0
\Bigr\},
\label{eq:QMI_set}
\end{align}
where
$\Phi_i:=Z_iZ_i^\top$,
$\Psi_i:=-Z_iX_i^{+\top}$,
$\Omega_i:=X_i^+X_i^{+\top}-\Delta_i\Delta_i^\top$.
Under Assumption \ref{ass:noise_bound}, the true pair
\([A_{\eta,ii}\;\;B_{\eta i}]\) belongs to \(\mathcal C_i\). In
Section~\ref{Data-Driven Impulsive Control}, a robust LMI will use this
set to synthesize \(K_i\) without selecting or identifying a
particular model.

\subsection{Spiking Control Realization}\label{Neuromorphic Impulsive}
The preceding data-driven synthesis determines the feedback gain $K_i$, but the resulting command $K_i\eta_i$ remains continuously valued and requires online modulation of the control amplitude. This leaves open how the data-designed feedback can be implemented through a spiking control architecture whose admissible control actions are restricted to a finite set of fixed impulse vectors. To bridge this gap, the feedback command is encoded into the firing times and channel selections of LIF units. The resulting spiking input acts only at firing instants, without continuously varying actuation between events. The key analytical question is whether this spiking realization preserves closed-loop stability and the prescribed $\epsilon$-NE accuracy of the nominal continuous-feedback system.

Following \cite{gerstner2002spiking,eilers2025stability}, LIF dynamics
are employed as online spike generators. For each agent
$i\in\mathcal N$, consider $L_i$ non-interacting scalar LIF units
indexed by $\mathcal L_i:=\{1,\ldots,L_i\}$. Each unit integrates a
feedback-dependent drive and emits a spike upon reaching its
threshold. The units have no direct state or reset coupling, although
their drives depend on the common signal $\eta_i$. In the
distributional sense, their dynamics are
\begin{equation}
 \dot z_{i,p}(t)
 =
-\lambda_{i,p} z_{i,p}(t)
 +g_{i,p}(\eta_i(t))-\theta_{i,p} s_{i,p}(t).
\label{eq:local_lif_dynamics}
\end{equation}
Here, \(z_{i,p}\) is the state of unit \(p\),
\(\lambda_{i,p}\ge0\) is its leakage rate, and
\(g_{i,p}(\eta_i)\ge0\) is the feedback-dependent drive integrated by
the unit. The threshold \(\theta_{i,p}>0\) determines when a spike is
emitted, and \(s_{i,p}\) denotes the resulting spike train.

Specifically, the \(k\)th event time \(t_{i,p}^k\) satisfies
\begin{equation}
z_{i,p}\bigl((t_{i,p}^k)^-\bigr)=\theta_{i,p},
\qquad i\in\mathcal N,\quad p\in\mathcal L_i.
\label{eq:event_triggering}
\end{equation}
At \(t_{i,p}^k\), unit \(p\) emits a spike. The corresponding
distributional term in \eqref{eq:local_lif_dynamics} resets its state
\begin{equation}
z_{i,p}\bigl((t_{i,p}^k)^+\bigr)
=
z_{i,p}\bigl((t_{i,p}^k)^-\bigr)-\theta_{i,p}
=0.
\label{eq:neuronal_reset}
\end{equation}
The neuronal states are initialized within their admissible ranges as
\begin{equation}\label{eq:admissible_neuronal_initialization}
0\le \Theta_i^{-1}z_i(0)<\mathbf{1},
\qquad i\in\mathcal{N},
\end{equation}
where the inequalities are understood componentwise. This range is invariant under the neuronal dynamics: since
$g_{i,p}(\eta_i)\ge 0$, a state at $z_{i,p}=0$ cannot decrease below zero, while a state reaching
$z_{i,p}=\theta_{i,p}$ fires and is immediately reset to zero. Consequently,
\begin{equation}
0\le z_{i,p}(t)\le \theta_{i,p},
\qquad i\in\mathcal{N},\; p\in\mathcal{L}_i.
\label{eq:neuronal_state_bound}
\end{equation}
The spikes emitted by unit $p$ are represented by the spike train
\begin{equation}
s_{i,p}(t)
=
\sum_{k\in\mathcal K_{i,p}}
\delta\!\left(t-t_{i,p}^k\right),
\qquad i\in\mathcal N,\quad p\in\mathcal L_i,
\label{eq:local_spike_train}
\end{equation}
where \(\delta(\cdot)\) is the Dirac delta and
\(\mathcal K_{i,p}\) indexes the events of unit \(p\). Let
\(s_i:=\operatorname{col}(s_{i,1},\ldots,s_{i,L_i})\) collect all spike
trains assigned to agent \(i\).

For compactness, define
\[
\begin{aligned}
z_i&:={\rm col}(z_{i,1},\ldots,z_{i,L_i}),
&g_i&:={\rm col}(g_{i,1},\ldots,g_{i,L_i}),\\
\Theta_i&:={\rm diag}(\theta_{i,1},\ldots,\theta_{i,L_i}),
&\Lambda_i&:={\rm diag}(\lambda_{i,1},\ldots,\lambda_{i,L_i}).
\end{aligned}
\]
Here, $\Theta_i$ and $\Lambda_i$ are the reset and leakage matrices, respectively. Stacking \eqref{eq:local_lif_dynamics} over $p\in\mathcal{L}_i$ gives
\begin{equation}
 \dot z_i(t)
 =
 -\Lambda_i z_i(t)
 +g_i(\eta_i(t))
 -\Theta_i s_i(t).
\label{eq:local_lif_dynamics_ar}
\end{equation}
Each unit is associated with a fixed impulse direction
\(\mu_{i,p}\in\mathbb R^{m_i}\). Define
\(M_i:=[\,\mu_{i,1}\ \cdots\ \mu_{i,L_i}\,]\). The system input is
\begin{equation}
u_i(t)=M_i s_i(t).
\label{eq:local_impulsive_input}
\end{equation}
An isolated spike of unit \(p\) therefore produces the jump
\[
x_i((t_{i,p}^k)^+)
=x_i((t_{i,p}^k)^-)+B_i\mu_{i,p},
\]
whereas no continuously valued control is applied between events.

Substituting \eqref{eq:local_impulsive_input} into the augmented agent
dynamics shows that the spike train \(s_i\) produces the impulsive term
\(B_{\eta i}M_i s_i\). The same spike train generates the reset term
\(-\Theta_i s_i\) in \eqref{eq:local_lif_dynamics_ar}. This matched
jump structure motivates the following auxiliary coordinate
\begin{equation}\label{eq:chi}
\eta_{ci}
:=
\eta_i
+
B_{\eta i}M_i\Theta_i^{-1}z_i.
\end{equation}
The coordinate \(\eta_{ci}\) is introduced for analysis and is
not evaluated in the controller implementation.

Substituting \eqref{eq:augmented_open_loop1} and
\eqref{eq:local_lif_dynamics_ar} into the derivative of
\eqref{eq:chi} gives
\begin{align}
\label{eq:al}
 \dot\eta_{ci}
 &=\dot\eta_i+B_{\eta i}M_i\Theta_i^{-1}\dot z_i\nonumber\\
 &=\sum_{j=1}^{N}A_{\eta,ij}\eta_j+B_{\eta i}M_i s_i
   +W_{\eta i}\nu_i\nonumber\\
 &\quad+B_{\eta i}M_i\Theta_i^{-1}
   \bigl(-\Lambda_i z_i+g_i(\eta_i)-\Theta_i s_i\bigr)\\
 &=\sum_{j=1}^{N}A_{\eta,ij}\eta_j+W_{\eta i}\nu_i+B_{\eta i}M_i\Theta_i^{-1}
   \bigl(g_i(\eta_i)-\Lambda_i z_i\bigr).\nonumber
\end{align}
The terms containing \(s_i\) cancel exactly. Thus, \(\eta_{ci}\) is
continuous even though \(\eta_i\) and \(z_i\) jump. To recover the
analogue feedback contribution in the continuous-time dynamics of this
coordinate, we impose the \emph{rate-equivalence condition}
\begin{equation}\label{eq:realization_nonint}
M_i\Theta_i^{-1}g_i(\eta_i)=K_i\eta_i,
\qquad i\in\mathcal N.
\end{equation}
Substitution into \eqref{eq:al} gives
\[
 \dot \eta_{ci}
 =\sum_{j=1}^{N}A_{\eta,ij}\eta_j+B_{\eta i}K_i\eta_i
   +W_{\eta i}\nu_i-B_{\eta i}M_i\Theta_i^{-1}\Lambda_i z_i.
\]
Hence, apart from a leakage-dependent term, the continuous-time dynamics contain
exactly the feedback contribution of the analogue reference
controller.
\begin{remark}\label{rem:nonnegative_encoding}
Since $g_i(\eta_i)\in\mathbb R_+^{L_i}$,
condition \eqref{eq:realization_nonint} can hold for all $\eta_i$ only if
$\operatorname{im}(K_i)\subseteq\operatorname{cone}(M_i)$.
For any data-designed gain \(K_i\), this condition can be satisfied by
assigning two neuronal units to each control-input channel. Set
\(L_i=2m_i\), choose a diagonal matrix
$\Xi_i=\operatorname{diag}(\xi_{i,1},\ldots,\xi_{i,m_i})\succ0$,
and define
$M_i=[\,\Xi_i\ -\Xi_i\,]\in\mathbb R^{m_i\times 2m_i}$.
For \(v_i:=K_i\eta_i\in\mathbb R^{m_i}\), let
$g_i(\eta_i)=
\Theta_i
\begin{bmatrix}
\begin{smallmatrix}
\Xi_i^{-1}[v_i]_+\\
\Xi_i^{-1}[-v_i]_+
\end{smallmatrix}
\end{bmatrix}
\in\mathbb R_{+}^{2m_i}$,
where \([v]_+\) denotes the componentwise positive part of \(v\).
Then
\begin{align*}
M_i\Theta_i^{-1}g_i(\eta_i)
&=
[\,v_i\,]_+
-
[\,-v_i\,]_+\\
&=
v_i
=
K_i\eta_i,
\end{align*}
and hence \eqref{eq:realization_nonint} is satisfied. Moreover,
\(\operatorname{cone}(M_i)=\mathbb R^{m_i}\), and \(g_i\) is
nonnegative and globally Lipschitz.
Condition \eqref{eq:realization_nonint} matches the nominal
feedback contribution in the auxiliary dynamics but does not assert
the pointwise identity \(M_i s_i(t)=K_i\eta_i(t)\), since \(M_i s_i\)
is distribution-valued whereas \(K_i\eta_i\) is continuously valued.
The remaining deviation from the nominal continuous-feedback dynamics
is the leakage-dependent term
\(-B_{\eta i}M_i\Theta_i^{-1}\Lambda_i z_i\), which is bounded because
\(z_i\) remains in its invariant threshold set.
\end{remark}
The preceding construction leads to the following data-driven synthesis problem.

\begin{problem}
Under Assumptions~\ref{ass:exosystem}-\ref{ass:noise_bound}, use only
the local offline data to construct a distributed spiking controller \(u_i=M_is_i\), without identifying the matrices in
\eqref{eq:agent_dynamics}, such that
\begin{enumerate}[(i)]
\item the resulting closed-loop system is forward complete and Zeno-free;
\item there exists a finite constant \( \epsilon_0\ge0\) such that, for
every \( \epsilon> \epsilon_0\), there exists
\(T_\epsilon\ge0\) for which \(y(t)\) is an
\(\epsilon\)-NE for all \(t\ge T_\epsilon\).
\end{enumerate}
\end{problem}

Figure~\ref{fig:NNDD} summarizes the overall implementation. Data are used
only in the offline gain design. During online operation, neighboring
outputs enter the internal models, the gains drive the neuronal
states, and only their spikes actuate the agents.
\begin{figure}[!htbp]
    \centering
    \includegraphics[scale=0.43]{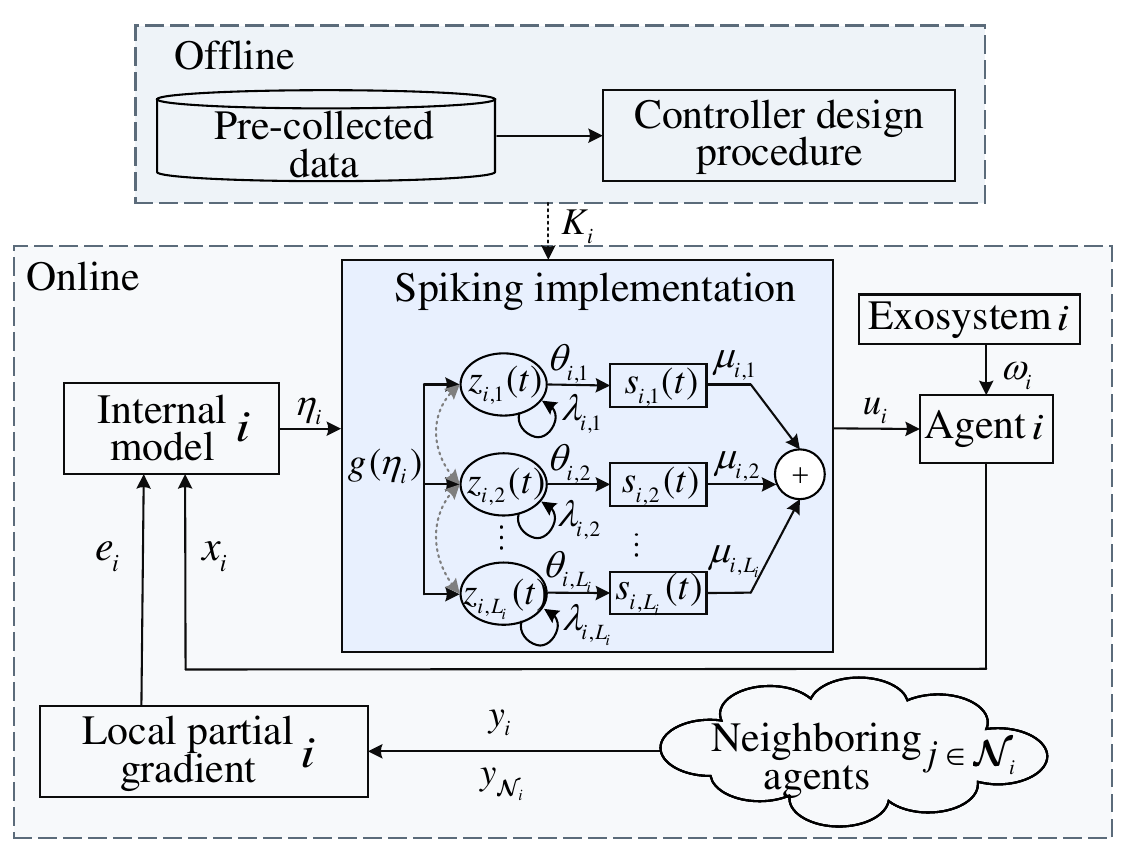}
    \caption{Data-driven spiking control architecture. Local
    offline data determine the analogue gains \(K_i\); online, these
    gains drive neuronal event generators whose fixed-weight spikes
    actuate the networked agents.}
    \label{fig:NNDD}
\end{figure}

\section{Non-interacting Spiking Realization: Synthesis and Guarantees}\label{Data-Driven Impulsive Control}
This section establishes the stability and equilibrium guarantees for
the non-interacting realization introduced in
Section~\ref{Neuromorphic Impulsive}. For compactness, define
\begin{align*}
    z&:=\operatorname{col}(z_1,\ldots,z_N),
    ~~~~~~~M:=\operatorname{blkdiag}(M_1,\ldots,M_N),\\
    \Lambda&:=\operatorname{blkdiag}(\Lambda_1,\ldots,\Lambda_N),
    ~\Theta:=\operatorname{blkdiag}(\Theta_1,\ldots,\Theta_N).
\end{align*}
Following the inter-event-time analysis in
\cite[Lemma~10 and Theorem~11]{eilers2025stability}, the next lemma
establishes the corresponding finite-horizon result for the
multi-agent setting considered here.
\begin{lemma}\label{lem:no_zeno_nonint}
Consider the non-interacting neuronal dynamics
\eqref{eq:local_lif_dynamics}, where each \(g_i\) is locally
Lipschitz. If, for some finite \(T>0\),
\begin{equation}
 \sup_{t\in[0,T)}\|\eta_i(t)\|<\infty,
 \qquad i\in\mathcal N,
 \label{eq:finite_horizon_bound}
\end{equation}
then every neuronal unit admits a positive lower bound on its inter-event times in $[0,T)$. Consequently, the aggregate event sequence has only finitely many events in $[0,T)$.
\end{lemma}

\noindent\textit{Proof:}
Fix \(i\in\mathcal N\) and \(p\in\mathcal{L}_i\). Since \(g_i\) is
locally Lipschitz and \(\eta_i\) is bounded on \([0,T)\), there exists
\(G_{i,p}(T)<\infty\) such that
\begin{equation}
 0\leq g_{i,p}(\eta_i(t))\leq G_{i,p}(T),
 \qquad t\in[0,T).
\end{equation}
Let $t_{i,p}^k<t_{i,p}^{k+1}$ be consecutive event times in $[0,T)$. The reset rule and the variation-of-constants formula give
\begin{equation}
 \theta_{i,p}
 =\int_{t_{i,p}^k}^{t_{i,p}^{k+1}}
 e^{-\lambda_{i,p}(t_{i,p}^{k+1}-\tau)}
 g_{i,p}(\eta_i(\tau))\,d\tau.
\end{equation}
Since $\lambda_{i,p}\geq0$,
\begin{equation}
 \theta_{i,p}
 \leq G_{i,p}(T)\bigl(t_{i,p}^{k+1}-t_{i,p}^k\bigr).
\end{equation}
Thus, if $G_{i,p}(T)>0$,
\begin{equation}
 t_{i,p}^{k+1}-t_{i,p}^k
 \geq\frac{\theta_{i,p}}{G_{i,p}(T)}>0.
\end{equation}
If $G_{i,p}(T)=0$, the unit cannot generate another event after a reset. Hence each unit generates only finitely many events in $[0,T)$. The conclusion follows because the numbers of agents and neuronal units are finite.

Lemma~\ref{lem:no_zeno_nonint} reduces Zeno-freeness to a
finite-horizon boundedness question. The next result answers that
question using an LMI expressed entirely in terms of the local data
matrices.

\begin{theorem}\label{thm:noninteracting}
Consider \eqref{eq:augmented_open_loop} interconnected with the
non-interacting neuronal dynamics
\eqref{eq:local_lif_dynamics_ar} under
Assumptions~\ref{ass:exosystem}-\ref{ass:noise_bound}. Suppose that,
for each \(i\in\mathcal N\), there exist
\(P_i=P_i^\top\succ0\) and \(H_i\) satisfying
\begin{equation}\label{eq:lmi_noninteracting}
\begin{bmatrix}
-\Omega_i & \Psi_i^\top-\Upsilon_i^\top\\
* & -\Phi_i
\end{bmatrix}
\prec0,\qquad
\Upsilon_i:=
\begin{bmatrix}
P_i\\
H_i
\end{bmatrix}.
\end{equation}
Define \(K_i:=H_iP_i^{-1}\) and
\(K:=\operatorname{blkdiag}(K_1,\ldots,K_N)\). Let each \(g_i\) be locally Lipschitz, nonnegative, and satisfy
\eqref{eq:realization_nonint}.
If \(0\le \Theta^{-1}z(0)<\mathbf 1\), then the closed-loop system under the non-interacting spiking controller is forward complete and Zeno-free. Moreover, there exists a finite constant \(\delta_0\ge0\) such that
\begin{equation}\label{eq:gradient_residual_bound}
\limsup_{t\to\infty}\|F(y(t))\|\le\delta_0.
\end{equation}
Let
$\epsilon_0
:=
\frac{\delta_0^2}{2\underline{\alpha}}$.
Then, for every \(\epsilon>\epsilon_0\), there exists
\(T_\epsilon\ge0\) such that \(y(t)\) is an
\(\epsilon\)-NE for all \(t\ge T_\epsilon\).
\end{theorem}

\noindent \textit{Proof:}
The proof proceeds from robust local stabilization to global
regulation and then to the spiking realization. First, define the
local analogue closed-loop matrix
$\widetilde A_{\eta,ii}:=A_{\eta,ii}+B_{\eta i}K_i$.
For the true local matrix
$\Gamma_i^\top=[\,A_{\eta,ii}\;\;B_{\eta i}\,]\in\mathcal C_i$,
pre- and postmultiplying \eqref{eq:lmi_noninteracting} by
$\operatorname{col}(I,-\Gamma_i)^\top$ and
$\operatorname{col}(I,-\Gamma_i)$, respectively, and using
\eqref{eq:QMI_set} yield
\begin{align}\label{eq:local_closed_loop_inequality}
 (A_{\eta,ii}P_i+B_{\eta i}H_i)
 +(A_{\eta,ii}P_i+B_{\eta i}H_i)^\top
 \prec0.
\end{align}
Since $K_i=H_iP_i^{-1}$, this inequality is equivalent to
\begin{equation}
 \widetilde A_{\eta,ii}P_i
 +P_i\widetilde A_{\eta,ii}^\top
 \prec0.
\label{eq:local_closed_loop_stability}
\end{equation}
Hence, \(\widetilde A_{\eta,ii}\) is Hurwitz for every
\(i\in\mathcal N\). Next, define the network-level analogue
closed-loop matrix
$\widetilde A_\eta:=A_\eta+B_\eta K$.
Since $B_\eta$ and $K$ are block diagonal, the diagonal blocks of
$\widetilde A_\eta$ are precisely $\widetilde A_{\eta,ii}$.
Moreover, under a topological ordering of the acyclic communication
digraph, $\widetilde A_\eta$ is block lower triangular. Therefore,
$\widetilde A_\eta$ is Hurwitz.

We next establish output regulation for this stable analogue model.
Assumption~\ref{ass:regulation_solvability} guarantees the
local regulator solvability conditions
\cite[Theorem~1.9]{huang2004nonlinear}. Together with the controllable
$p_i$-copy internal models and the Hurwitz property of
$\widetilde A_\eta$, the regulator equation argument in
\cite[Theorem~1]{guo2021linear} guarantees the existence of a matrix
$\Sigma$ satisfying
\begin{subequations}\label{eq:regulator_equations_simplified}
\begin{align}
\Sigma\bar E=\widetilde A_\eta\Sigma+W_\eta,\label{eq:regulator_equations_simplified1}\\
\qquad
0=C_\eta\Sigma+Q_\eta \label{eq:regulator_equations_simplified2}.
\end{align}
\end{subequations}
We now transfer this property to the spiking implementation.
Let $r:=\Theta^{-1}z$ and
$\mathcal R:=[0,1]^{L_{\mathrm{tot}}}$, where
$L_{\mathrm{tot}}:=\sum_{i=1}^N L_i$. The initialization \eqref{eq:admissible_neuronal_initialization} and the reset rule imply $r(t)\in\mathcal R$.
Stacking the auxiliary variables introduced in \eqref{eq:chi} gives
$\eta_c:=\eta+B_\eta Mr$. The cancellation established in Section \ref{Neuromorphic Impulsive} makes $\eta_c$ continuous at the event times.  

Using \eqref{eq:realization_nonint} and
$\eta=\eta_c-B_\eta Mr$, its continuous dynamics can be written as
\begin{equation}\label{eq:eta}
\dot\eta_c
=
\widetilde A_\eta\eta_c+W_\eta\nu+E_0r ,
\end{equation}
where
\[
E_0:=-\widetilde A_\eta B_\eta M
-B_\eta M\Theta^{-1}\Lambda\Theta.
\]
Thus, in the auxiliary coordinate, the spiking implementation is represented by the stable analogue closed loop driven by the bounded signal \(E_0r\).

Define $\chi_c:=\eta_c-\Sigma\nu$. From \eqref{eq:regulator_equations_simplified},
\begin{subequations} \label{eq:regulator_al}
\begin{align}
 \dot\chi_c&=\widetilde A_\eta\chi_c+E_0r,\label{eq:regulator_al1}\\
 \eta&=\chi_c+\Sigma\nu-B_\eta Mr \label{eq:regulator_al2}.   
\end{align}
\end{subequations} 
Since \(r\) is bounded, \eqref{eq:regulator_al} implies that
\(\chi_c\), and hence \(\eta\), is bounded on every compact time
interval. Lemma~\ref{lem:no_zeno_nonint} therefore excludes finite
accumulation of event times. The same bounds preclude finite escape,
so every solution is forward complete and Zeno-free.

It remains to quantify the regulation residual. Since
$\widetilde A_\eta$ is Hurwitz, for any
$Q_c=Q_c^{\top}\succ0$, there exists
$P_c=P_c^{\top}\succ0$ satisfying
\begin{equation}\label{eq:network_lyapunov}
P_c\widetilde A_\eta
+\widetilde A_\eta^{\top}P_c
=-Q_c.
\end{equation}
Let
$V:=\chi_c^{\top}P_c\chi_c$,
$\rho_0:=\max_{r\in\mathcal R}\|P_cE_0r\|$.
The maximum is finite because $\mathcal R$ is compact. Moreover,
$\chi_c$ is continuous at every firing instant, and hence $V$ has no
jumps. Along the continuous flows,
\begin{equation}\label{eq:lyapunov_derivative_bound}
\dot V
\le
-\lambda_{\min}(Q_c)\|\chi_c\|^2
+2\rho_0\|\chi_c\|.
\end{equation}

Define
$\beta_0
:=
\frac{2\rho_0}{\lambda_{\min}(Q_c)}$.
Whenever
$V>\lambda_{\max}(P_c)\beta_0^2$, one has
$\|\chi_c\|>\beta_0$, and
\eqref{eq:lyapunov_derivative_bound} gives $\dot V<0$. Consequently,
\[
\limsup_{t\to\infty}V(t)
\le
\lambda_{\max}(P_c)\beta_0^2.
\]
Using
$V\ge\lambda_{\min}(P_c)\|\chi_c\|^2$ yields
\begin{equation}\label{eq:chi_ultimate_bound}
\limsup_{t\to\infty}\|\chi_c(t)\|
\le
\frac{2\sqrt{\kappa(P_c)}\,\rho_0}
     {\lambda_{\min}(Q_c)}.
\end{equation}

Using the regulator equation and the relation
$\eta=\chi_c+\Sigma\nu-B_\eta Mr$ gives
\[
e
=
C_\eta\eta+Q_\eta\nu
=
C_\eta(\chi_c-B_\eta Mr).
\]
Define
$\delta_M
:=
\max_{r\in\mathcal R}\|B_\eta Mr\|<\infty$.
Combining this relation with
\eqref{eq:chi_ultimate_bound} gives
\begin{equation}\label{eq:explicit_delta_0}
\begin{aligned}
\limsup_{t\to\infty}\|e(t)\|
&\le
\|C_\eta\|
\left(
\frac{2\sqrt{\kappa(P_c)}\,\rho_0}
     {\lambda_{\min}(Q_c)}
+\delta_M
\right)\\
&=:\delta_0.
\end{aligned}
\end{equation}
Since $e=F(y)$, the claimed pseudo-gradient bound follows.

Finally, let \(\epsilon>\epsilon_0\) and define
$r_\epsilon:=\sqrt{2\underline{\alpha}\epsilon}$.
Since
$\epsilon>\epsilon_0=\frac{\delta_0^2}{2\underline{\alpha}}$,
we have \(r_\epsilon>\delta_0\). Hence, by
\eqref{eq:gradient_residual_bound}, there exists
\(T_\epsilon\ge0\) such that
\[
\|F(y(t))\|<r_\epsilon,
\qquad t\ge T_\epsilon.
\]
Applying \eqref{eq:aggregate_regret_bound} gives
\begin{align}
&\max_{i\in\mathcal N}
\left\{
J_i(y_i(t),y_{\mathcal N_i}(t))
-
\inf_{\bar y_i\in\mathbb R^{p_i}}
J_i(\bar y_i,y_{\mathcal N_i}(t))
\right\}\nonumber\\
&\qquad
\le
\frac{\|F(y(t))\|^2}{2\underline{\alpha}}
<\epsilon.
\end{align}
Therefore, \(y(t)\) is an \(\epsilon\)-NE for all
\(t\ge T_\epsilon\).

\begin{remark}
Theorem~\ref{thm:noninteracting} separates what is
\emph{data-computable} from what is \emph{existential}. The gain
\(K_i=H_iP_i^{-1}\) is obtained directly from the local data through
\eqref{eq:lmi_noninteracting}, without identifying the agent
dynamics. The bound \(\delta_0\) in
\eqref{eq:explicit_delta_0}, however, depends
on the actual closed-loop matrices through \(P_c\), \(E_0\), and
\(\delta_M\). The theorem therefore guarantees a finite
\(\epsilon_0\), but the particular analytical bound derived here
cannot yet be evaluated from the offline data alone. A less
conservative, fully data-computable equilibrium-accuracy certificate
is an important direction for future work.
\end{remark}

\section{Connected Neuronal Units: Reset Coupling and Guarantees}\label{Extension to Connected Neuronal Units}
The non-interacting spiking realization in Section \ref{Data-Driven Impulsive Control} employs unit-wise resets. The paired construction in Remark~\ref{rem:nonnegative_encoding} provides an explicit realization for any data-designed gain \(K_i\), but uses \(L_i=2m_i\) neuronal units. This section considers connected neuronal units as an alternative realization based on a more compact set of fixed impulse directions. The cone condition imposed below admits choices with as few as \(L_i=m_i+1\) directions, which are generally nonorthogonal.

The inner product \(\mu_{i,p}^{\top}K_i\eta_i\) measures the
alignment between the impulse direction and the nominal feedback
command.  For nonorthogonal directions, independent resets may cause
redundant firing; following \cite[Sec.~III-C]{eilers2025stability},
we therefore introduce reset coupling among the neuronal units
\begin{equation}\label{eq:connected_dynamics}
    \dot z_i=-\Lambda_i z_i+M_i^\top K_i\eta_i-M_i^\top M_i s_i,  \qquad i\in\mathcal{N},
\end{equation}
where \(s_i\) collects the Dirac spike trains. If unit \(p\) fires,
then
\[
 z_{i,q}^{+}=z_{i,q}^{-}-\mu_{i,q}^{\top}\mu_{i,p},
 \qquad q\in\mathcal{L}_i.
\]
Thus, a spike modifies all neuronal states according to the pairwise
inner products of the impulse directions. In particular, it inhibits
aligned directions and excites oppositely oriented directions.

The self-reset associated with unit \(p\) is
\(\|\mu_{i,p}\|^2\). Accordingly, define
\begin{equation}\label{eq:connected_threshold}
d_{M_i}
:=
\operatorname{col}
\bigl(
\|\mu_{i,1}\|^2,\ldots,\|\mu_{i,L_i}\|^2
\bigr).
\end{equation}
When $z_{i,p}$ reaches the corresponding threshold
$d_{M_i,p}=\|\mu_{i,p}\|^2$ from below during continuous
evolution, unit $p$ emits a spike and its own state is
reset to zero. A reset may also make another unit active and thereby initiate a cascade of spikes at the same physical time. Unlike the
non-interacting realization, the drive \(M_i^\top K_i\eta_i\) is not
necessarily nonnegative; hence, the individual neuronal states need
not remain nonnegative.

The following assumption specifies the processing of reset cascades and the associated no-accumulation condition.
\begin{assumption}\label{ass:connected}
For each \(i\in\mathcal N\), the matrix \(M_i\) has no zero
column and satisfies
\begin{equation}\label{eq:cone_connected}
 \operatorname{cone}(M_i)=\mathbb R^{m_i}.
\end{equation}
Moreover, \(\Lambda_i=\lambda_iI_{L_i}\) for some
\(\lambda_i\ge0\). At an event time, an active unit \(p\) satisfying
\(z_{i,p}\ge d_{M_i,p}\) may be selected to spike, producing
\[
 z_i^+=z_i^- -M_i^\top M_i e_p,
\]
where \(e_p\in\mathbb R^{L_i}\) denotes the \(p\)th standard basis vector.
The active set is updated after each reset. The resulting cascade
terminates after finitely many resets and leaves no active unit. Its
terminal state defines the right-continuous value of \(z_i\) and
satisfies
\begin{equation}\label{eq:upper_bound_connected}
 z_{i,p}(t)< d_{M_i,p},
 \qquad p\in\mathcal L_i.
\end{equation}
For every finite \(T>0\), boundedness of \(\eta\) on \([0,T)\)
implies that only finitely many distinct event times occur in
\([0,T)\).
\end{assumption}

The nonzero-column condition makes every threshold positive. The cone
condition and the upper bound in \eqref{eq:upper_bound_connected}
together also imply a lower bound on every neuronal state, as shown
next.

\begin{lemma}\label{lem:connected_bound}
Consider \eqref{eq:connected_dynamics} with \(z_i(0)=0\) and the
thresholds \eqref{eq:connected_threshold}. Under
Assumption~\ref{ass:connected}, there exists a compact set
\(\mathcal Z_{M_i}\subset\mathbb{R}^{L_i}\) such that
\(z_i(t)\in\mathcal Z_{M_i}\) throughout the domain of every solution.
\end{lemma}

\noindent \textit{Proof:}
By \eqref{eq:cone_connected}, $-\mu_{i,p}\in\operatorname{cone}(M_i)$. Hence, for each $p\in\mathcal L_i$, there exists $\beta_{i,p}\in\mathbb R_+^{L_i}$ such that
\begin{equation}
M_i\beta_{i,p}=-\mu_{i,p}.
\end{equation}
Define \(w_i:=\sum_{p=1}^{L_i}(e_p+\beta_{i,p})\). Then
\(w_i>0\) componentwise and \(M_iw_i=0\). Premultiplying
\eqref{eq:connected_dynamics} by \(w_i^\top\) eliminates the drive and
reset terms containing \(M_i^\top\). Since
\(\Lambda_i=\lambda_iI_{L_i}\),
\begin{equation}
\frac{d}{dt}w_i^\top z_i
=-\lambda_i w_i^\top z_i,
\end{equation}
in the distributional sense. Hence $z_i(0)=0$ gives
\begin{equation}\label{eq:wz_zero}
w_i^\top z_i(t)=0.
\end{equation}
Combining \eqref{eq:wz_zero} with the
upper-bound condition \eqref{eq:upper_bound_connected}, for each
\(p\in\mathcal L_i\),
\begin{equation}
w_{i,p}z_{i,p}(t)
=-\sum_{q\ne p}w_{i,q}z_{i,q}(t)
\ge
-\sum_{q\ne p}w_{i,q}d_{M_i,q},
\end{equation}
and therefore
\begin{equation}
z_{i,p}(t)\ge
-\frac{1}{w_{i,p}}\sum_{q\ne p}w_{i,q}d_{M_i,q}.
\end{equation}
Together with \eqref{eq:upper_bound_connected}, this yields $z_i(t)\in\mathcal Z_{M_i}$, where
\begin{equation}\label{eq:ZMi}
\mathcal Z_{M_i}:=
\prod_{p=1}^{L_i}
\bigg[
-\frac{1}{w_{i,p}}\sum_{q\ne p}w_{i,q}d_{M_i,q},
\;d_{M_i,p}
\bigg],
\end{equation}
The set \(\mathcal Z_{M_i}\) is compact, proving the claim.

The following result establishes that the gain synthesized in
Section \ref{Data-Driven Impulsive Control} retains its stability and equilibrium guarantees under the spiking realization with connected neuronal units. In contrast to
Theorem~\ref{thm:noninteracting}, the main issues are the boundedness
of the reset-coupled neuronal states and the cancellation of all jumps
generated by a reset cascade.

\begin{theorem}\label{thm:connected}
Consider \eqref{eq:augmented_open_loop} with the impulsive input
\(u_i=M_is_i\) and the connected neuronal dynamics
\eqref{eq:connected_dynamics} under
Assumptions~\ref{ass:exosystem}-\ref{ass:noise_bound} and
Assumption~\ref{ass:connected}. Suppose that, for each \(i\in\mathcal N\), there exist
\(P_i=P_i^\top\succ0\) and \(H_i\) satisfying
\eqref{eq:lmi_noninteracting}.
Let \(K_i:=H_iP_i^{-1}\) and
\(K:=\operatorname{blkdiag}(K_1,\ldots,K_N)\). 
If \(z(0)=0\), every solution generated by the spiking rule
is forward complete and Zeno-free. Moreover, there exists a finite constant \(\delta_c\ge0\) such that
\begin{equation}\label{eq:connected_F_bound}
\limsup_{t\to\infty}\|F(y(t))\|\le\delta_c.
\end{equation}
Let
$\epsilon_{c,0}:=\frac{\delta_c^2}{2\underline{\alpha}}$.
Then, for every \(\epsilon>\epsilon_{c,0}\), there exists
\(T_\epsilon\ge0\) such that \(y(t)\) is an
\(\epsilon\)-NE for all \(t\ge T_\epsilon\).
\end{theorem}

\noindent \textit{Proof:}
The data-driven part of the argument is unchanged. As shown in
Theorem~\ref{thm:noninteracting}, feasibility of
\eqref{eq:lmi_noninteracting} makes
\(\widetilde A_\eta:=A_\eta+B_\eta K\) Hurwitz and guarantees a
solution \(\Sigma\) of
\eqref{eq:regulator_equations_simplified}.

The cone condition in Assumption~\ref{ass:connected} implies that
every \(M_i\) has full row rank. Define
\[
R_{M_i}:=(M_iM_i^\top)^{-1}M_i,~~
R_M:=\operatorname{blkdiag}(R_{M_1},\ldots,R_{M_N}).
\]
Then
\begin{equation}\label{eq:RM_identity}
R_MM^\top=I,
\qquad
R_MM^\top M=M,
\end{equation}
and consequently
\begin{equation}\label{eq:connected_feedback_recovery}
R_MM^\top K\eta=K\eta.
\end{equation}
Introduce the auxiliary coordinate
\begin{equation}\label{eq:connected_auxiliary_coordinate}
 \eta_c:=\eta+B_\eta R_Mz.
\end{equation}
At a firing instant, let $\Delta v:=v^+-v^-$ denote the
jump of any state $v$. If a neuronal unit spikes, with $e_p$
denoting the corresponding standard basis vector in the stacked
coordinates, then
$\Delta\eta=B_\eta Me_p$,
 $\Delta z=-M^\top Me_p$.
It follows from \eqref{eq:RM_identity} that
\[
 \Delta\eta_c
 =
 B_\eta Me_p-B_\eta R_MM^\top Me_p
 =0.
\]
Thus, \(\eta_c\) is unchanged by each reset and, consequently, by
every finite reset cascade. This establishes jump cancellation for the reset-coupled realization.

Along continuous flows, 
\begin{equation}\label{eq:connected_eta_c_dynamics}
\dot\eta_c = \widetilde A_\eta\eta_c+W_\eta\nu+E_cz,
\end{equation}
where
$ E_c:= -\bigl(\widetilde A_\eta B_\eta R_M+B_\eta R_M\Lambda\bigr)$.
With \(\chi_c:=\eta_c-\Sigma\nu\),
\eqref{eq:regulator_equations_simplified} and
\eqref{eq:connected_eta_c_dynamics} give
\begin{subequations}\label{eq:connected_chi_c}
\begin{align}
\dot\chi_c&=\widetilde A_\eta\chi_c+E_cz,\\
\eta&=\chi_c+\Sigma\nu-B_\eta R_Mz.
\end{align}
\end{subequations}
Lemma~\ref{lem:connected_bound} ensures that
$
z(t)\in\mathcal Z_M
:=
\mathcal Z_{M_1}\times\cdots\times\mathcal Z_{M_N}$
where \(\mathcal Z_M\) is compact. Since \(\nu\) is bounded on finite
intervals, \eqref{eq:connected_chi_c} shows that \(\chi_c\) and
\(\eta\) are also bounded on every finite interval.
Assumption~\ref{ass:connected} then excludes accumulation of event
times. Because only finitely many resets occur at any one event time
and the state cannot escape in finite time, every maximal solution can
be continued indefinitely. The closed loop is therefore forward
complete and Zeno-free.

The remaining regulation argument has the same form as that in
Theorem~\ref{thm:noninteracting}, with the bounded terms
\(E_0r\) and \(B_\eta Mr\) replaced by \(E_cz\) and
\(B_\eta R_Mz\), respectively. In particular, for any
\(Q_c=Q_c^\top\succ0\), let \(P_c=P_c^\top\succ0\) solve
\[
 P_c\widetilde A_\eta
 +\widetilde A_\eta^\top P_c
 =-Q_c.
\]
Define $\rho_c:=\max_{\bar z\in\mathcal Z_M}
 \|P_cE_c\bar z\|$, and 
$\delta_R:=
 \max_{\bar z\in\mathcal Z_M}$ $
 \|B_\eta R_M\bar z\|$.
Both constants are finite because \(\mathcal Z_M\) is compact.
Applying the Lyapunov argument in the proof of
Theorem~\ref{thm:noninteracting} gives
\[
 \limsup_{t\to\infty}\|\chi_c(t)\|
 \le
 \frac{2\sqrt{\kappa(P_c)}\,\rho_c}
      {\lambda_{\min}(Q_c)}.
\]
Moreover, the regulator equations yield
$e=C_\eta\bigl(\chi_c-B_\eta R_Mz\bigr)$.
Consequently,
\begin{align}\label{eq:connected_regulation_bound}
 \limsup_{t\to\infty}\|F(y(t))\|
 &=
 \limsup_{t\to\infty}\|e(t)\|\\
 &\le
 \|C_\eta\|
 \left(
 \frac{2\sqrt{\kappa(P_c)}\,\rho_c}
      {\lambda_{\min}(Q_c)}
 +\delta_R
 \right)
 =:\delta_c.\nonumber
\end{align}
Finally, let \(\epsilon>\epsilon_{c,0}\). Since
$\epsilon>\epsilon_{c,0}=\frac{\delta_c^2}{2\underline{\alpha}}$,
we have
$\sqrt{2\underline{\alpha}\epsilon}>\delta_c$.
It follows from \eqref{eq:connected_F_bound} that there exists
\(T_\epsilon\ge0\) such that
\[
\|F(y(t))\|
<
\sqrt{2\underline{\alpha}\epsilon},
\qquad t\ge T_\epsilon.
\]
Applying \eqref{eq:aggregate_regret_bound} gives
\[
\max_{i\in\mathcal N}
\left\{
J_i(y_i(t),y_{\mathcal N_i}(t))
-
\inf_{\bar y_i\in\mathbb R^{p_i}}
J_i(\bar y_i,y_{\mathcal N_i}(t))
\right\}
<\epsilon
\]
for all \(t\ge T_\epsilon\). Hence, \(y(t)\) is an
\(\epsilon\)-NE for all \(t\ge T_\epsilon\). This completes the proof.
\begin{remark}
Theorems~\ref{thm:noninteracting} and
\ref{thm:connected} use the same data-designed gain but concern different spiking realizations. In the non-interacting case,
boundedness follows directly from the componentwise threshold resets.
For connected neuronal units, reset coupling may generate negative
neuronal states and simultaneous reset cascades; boundedness instead
follows from Lemma~\ref{lem:connected_bound}, while jump cancellation
requires the left inverse \(R_M\). The common form of the final
regulation bound reflects the fact that both realizations recover the
same stable analogue closed loop subject to different bounded
realization residuals.
\end{remark}

\section{Numerical Study: Spacecraft Formation Reconfiguration}\label{sec:simulation}
We consider an in-plane formation reconfiguration problem involving
seven spacecraft with distinct masses in low Earth orbit. For small
relative separations about a circular reference orbit, their dynamics
in the local-vertical/local-horizontal frame are described by the
controlled Clohessy--Wiltshire--Hill equations
\cite{clohessy1960terminal}
\begin{align}
\ddot{\rho}_{x,i}-2n\dot{\rho}_{y,i}-3n^2\rho_{x,i}
&=\frac{u_{x,i}}{m_i}+\omega_{x,i},\qquad i=1,\ldots,7,\notag\\
\ddot{\rho}_{y,i}+2n\dot{\rho}_{x,i}
&=\frac{u_{y,i}}{m_i}+\omega_{y,i},
\label{eq:sim-agent-dynamics}
\end{align}
where \(\rho_{x,i}\) and \(\rho_{y,i}\) are the radial and along-track
positions, \(u_{x,i}\) and \(u_{y,i}\) are the control forces,
\(\omega_{x,i}\) and \(\omega_{y,i}\) are acceleration disturbances, \(m_i\) is
the spacecraft mass, and \(n\) is the orbital mean motion.

Let $\rho_i:={\rm col}(\rho_{x,i},\rho_{y,i})$, $\upsilon_i:={\rm col}(\dot\rho_{x,i},\dot\rho_{y,i})$, $u_i:={\rm col}(u_{x,i},u_{y,i})$, and $\omega_i:={\rm col}(\omega_{x,i},\omega_{y,i})$. With $x_i:={\rm col}(\rho_i,\upsilon_i)$ and $y_i:=\rho_i$, \eqref{eq:sim-agent-dynamics} takes the form in \eqref{eq:agent_dynamics}, with
\begin{align*}
A_{i}&=
\scalebox{0.85}{$\begin{bmatrix}
0&0&1&0\\
0&0&0&1\\
3n^2&0&0&2n\\
0&0&-2n&0
\end{bmatrix}$},
&
B_i&=
\scalebox{0.85}{$\begin{bmatrix}
0&0\\
0&0\\
1/m_i&0\\
0&1/m_i
\end{bmatrix}$},\\
W_i&=
\scalebox{0.85}{$\begin{bmatrix}
0&0&1&0\\
0&0&0&1
\end{bmatrix}^{\top}$},
&
C_i&=
\scalebox{0.85}{$\begin{bmatrix}
1&0&0&0\\
0&1&0&0
\end{bmatrix}$}.
\end{align*}
The common reference orbit is circular with an altitude of \(h=500\,\mathrm{km}\) and a mean motion of \(n=1.1068\times10^{-3}\,\mathrm{s}^{-1}\). The masses of the seven spacecraft are $(m_1,\ldots,m_7)=(10.5,11,11.5,12,12.5,13,13.5)\,\mathrm{kg}$.
The directed information graph is shown in Fig.~\ref{fig:infor}.
\begin{figure}[t]
    \centering
    \includegraphics[scale=0.5]{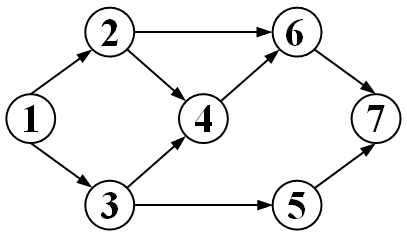}
    \caption{Directed information graph for the seven-spacecraft network.}
    \label{fig:infor}
\end{figure}

The acceleration disturbance \(\omega_i\) is generated by
$\dot \omega_i=\begin{bmatrix}
\begin{smallmatrix}
0&n\\-n&0\end{smallmatrix}\end{bmatrix} \omega_i$.
Because the extended exosystem also contains the constant offset in
the pseudo-gradient, its minimal polynomial is \(s(s^2+n^2)\). We use
the two-copy internal model
\[
\begin{aligned}
\Pi_i&=\operatorname{blkdiag}(\pi,\pi),&
\widehat{\Pi}_i&=\operatorname{blkdiag}(\widehat{\pi},\widehat{\pi}),\\
\pi&=
\begin{bmatrix}
0&1&0\\
0&0&1\\
0&-n^2&0
\end{bmatrix},&
\widehat{\pi}&=\operatorname{col}(0,0,1).
\end{aligned}
\]
The pair \((\pi,\widehat{\pi})\) is controllable and reproduces the
zero and harmonic modes required by the regulation model.
Each spacecraft minimizes
\begin{align}
    J_i(y_i,y_{\mathcal N_i})
    =
    \|y_i-r_i\|^2+
    \sum_{j\in\mathcal N_i}\|y_i-y_j\|^2,
    \label{eq:sim-cost}
\end{align}
where \(r_i\) is its prescribed relative position. The initial and
prescribed positions, in meters, are
\begin{align*}
 &[\,y_1(0)\ \cdots\ y_7(0)\,]
 =\scalebox{0.85}{$\begin{bmatrix}
 -5&4&-6&5&-3&3&8\\
 5&3&5&5&-5&-5&-5
 \end{bmatrix}$},\\
& [\,r_1\ \cdots\ r_7\,]
 =\scalebox{0.85}{$\begin{bmatrix}
 0&-10&15&-10&10&-9&9\\
 8&10&12&0&0&-9&-9
 \end{bmatrix}$}.
\end{align*}
The graph is acyclic, and each cost is strictly convex in \(y_i\).
Equation~\eqref{eq:unique_NE} gives the unique NE
\begin{align*}
 \scalebox{0.85}{$\setlength{\arraycolsep}{2.2pt}
 [\,y_1^\star\ \cdots\ y_7^\star\,]
 =\begin{bmatrix}
 0&-5&7.5&-2.5&8.75&-5.5&4.083\\
 8&9&10&6.333&5&2.111&-0.630
 \end{bmatrix}$}.
\end{align*}
For each spacecraft, the offline local data-acquisition procedure in
Section~\ref{Local Data} is conducted over \(8\,\mathrm{s}\), using \(T=80\)
samples with sampling interval \(h_d=0.1\,\mathrm{s}\). On each
sampling interval, a zero-order-held probing input is applied, with
each component independently drawn from
\([-50,50]\,\mathrm{N}\). The initial position and velocity components
are drawn from \([-10,10]\,\mathrm{m}\) and
\([-0.02,0.02]\,\mathrm{m/s}\), respectively, while the internal-model
state is initialized at zero. Each component of the initial disturbance state
\(\omega_i(0)\) is sampled once from
\([-3\times10^{-4},3\times10^{-4}]\), and the subsequent
disturbance trajectory is generated by
\(\dot\omega_i=E_i\omega_i\). Each component of the
derivative-approximation error \(d_i(t_\ell)\) is independently
sampled from \([-10^{-3},10^{-3}]\) at the data-acquisition
instants. The physical matrices are used
only by the simulator to generate the data and closed-loop
trajectories; controller synthesis uses only the locally collected
data and the prescribed uncertainty bound. For every \(i\in\mathcal N\),
the resulting matrix
\(Z_i\in\mathbb R^{12\times80}\) satisfies
\(\operatorname{rank}(Z_i)=12\). The matrix \(\Delta_i\) is selected
from known bounds on the aggregate disturbance and derivative error
such that Assumption~\ref{ass:noise_bound} holds. The local LMI
\eqref{eq:lmi_noninteracting} then yields
\(K_i=H_iP_i^{-1}\), which is used in the nominal analogue controller
and both spiking realizations.

The scalar \(\gamma>0\) specifies the common scale of the fixed impulse
directions, and \(\gamma=4\) is used throughout the simulations. For
the non-interacting realization, let
$L_i=4$, $M_i=\gamma[\,I_2\ -I_2\,]$, $\Theta_i=\gamma I_4$, and
$g_i(\eta_i)=\operatorname{col}\bigl([K_i\eta_i]_+,[-K_i\eta_i]_+\bigr).$
It follows directly that
\(M_i\Theta_i^{-1}g_i(\eta_i)=K_i\eta_i\), so that
\eqref{eq:realization_nonint} is satisfied. For the connected
realization, only three impulse directions are used
$L_i=3$,
$M_i=\gamma
\begin{bmatrix}
\begin{smallmatrix}
1&0&-1\\
0&1&-1
\end{smallmatrix}
\end{bmatrix}$.
This matrix has no zero column and satisfies
\(\operatorname{cone}(M_i)=\mathbb R^2\). Its threshold vector is
$d_{M_i}
=
\operatorname{col}
\bigl(
\|\mu_{i,1}\|^2,
\|\mu_{i,2}\|^2,
\|\mu_{i,3}\|^2
\bigr)
=
\gamma^2\operatorname{col}(1,1,2).$
The dynamics \eqref{eq:connected_dynamics} and the threshold rule
\eqref{eq:connected_threshold} are therefore applied, with
simultaneously active units processed according to
Assumption~\ref{ass:connected}. Both realizations use
\(z_i(0)=0\) and \(\Lambda_i=0.2I_{L_i}\).

Figures~\ref{fig:sim-analogue}--\ref{fig:sim-connected} show the
decision trajectories over \(t\in[0,20]\,\mathrm{s}\), where circles
and crosses denote the initial positions and the Nash equilibrium,
respectively. The nominal analogue controller drives the outputs to
\(y^\star\). The spiking realizations with non-interacting and connected neuronal units drive the outputs to neighborhoods of the same Nash configuration using only fixed-weight spikes. Although all three implementations use the same gains \(K_i\), the
neuronal trajectories need not coincide with the analogue trajectories
because rate equivalence does not imply the pointwise identity
\(M_is_i(t)=K_i\eta_i(t)\). The observed residuals are consistent with
the leakage-dependent practical-equilibrium bounds established in
Theorems~\ref{thm:noninteracting} and~\ref{thm:connected}.

\begin{figure}[!htbp]
    \centering
    \includegraphics[scale=0.5]{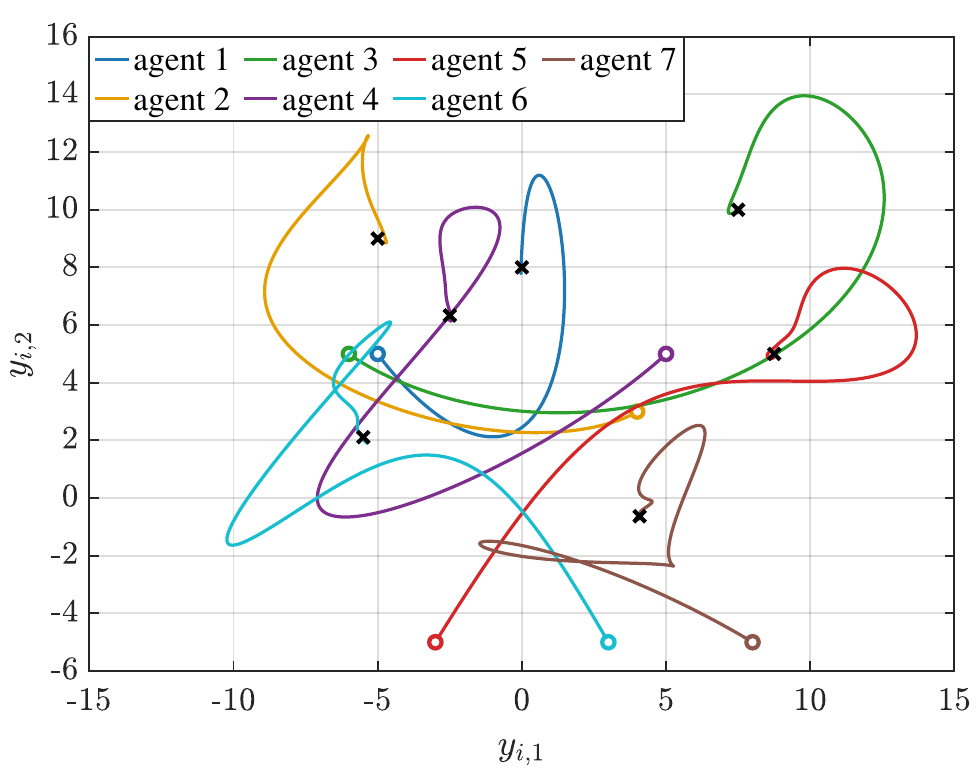}
    \caption{Decision-output trajectories under the continuously
    valued analogue reference controller.}
    \label{fig:sim-analogue}
\end{figure}
\begin{figure}[!htbp]
    \centering
    \includegraphics[scale=0.5]{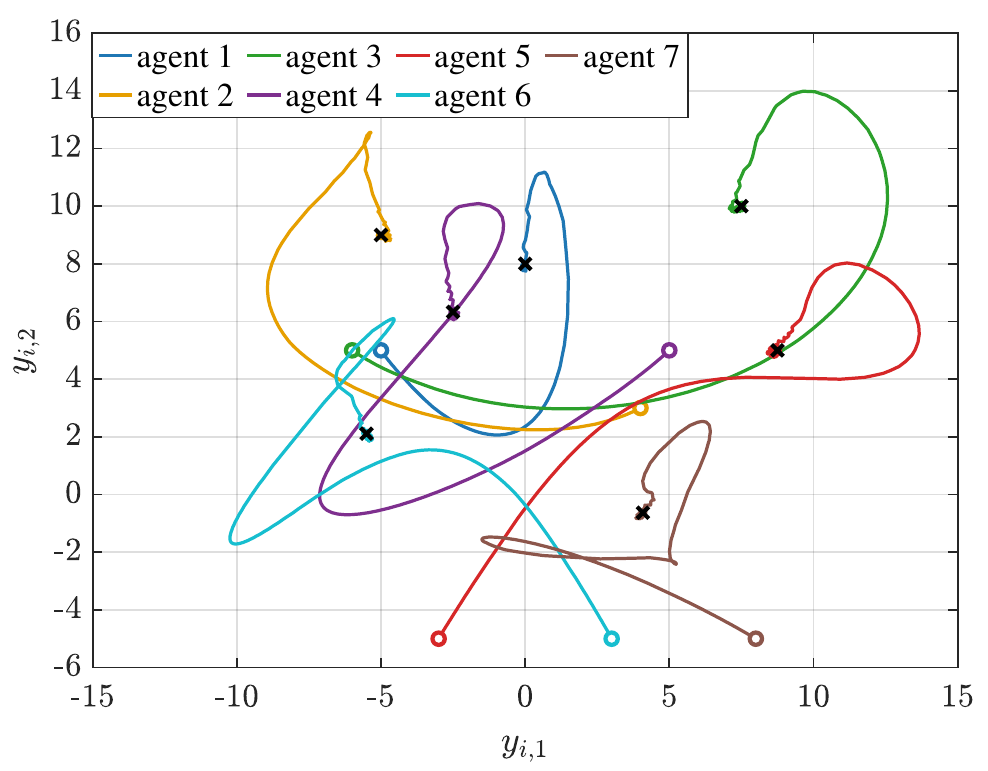}
    \caption{Decision-output trajectories under the non-interacting spiking controller.}
    \label{fig:sim-noninteracting}
\end{figure}

\begin{figure}[!htbp]
    \centering
    \includegraphics[scale=0.5]{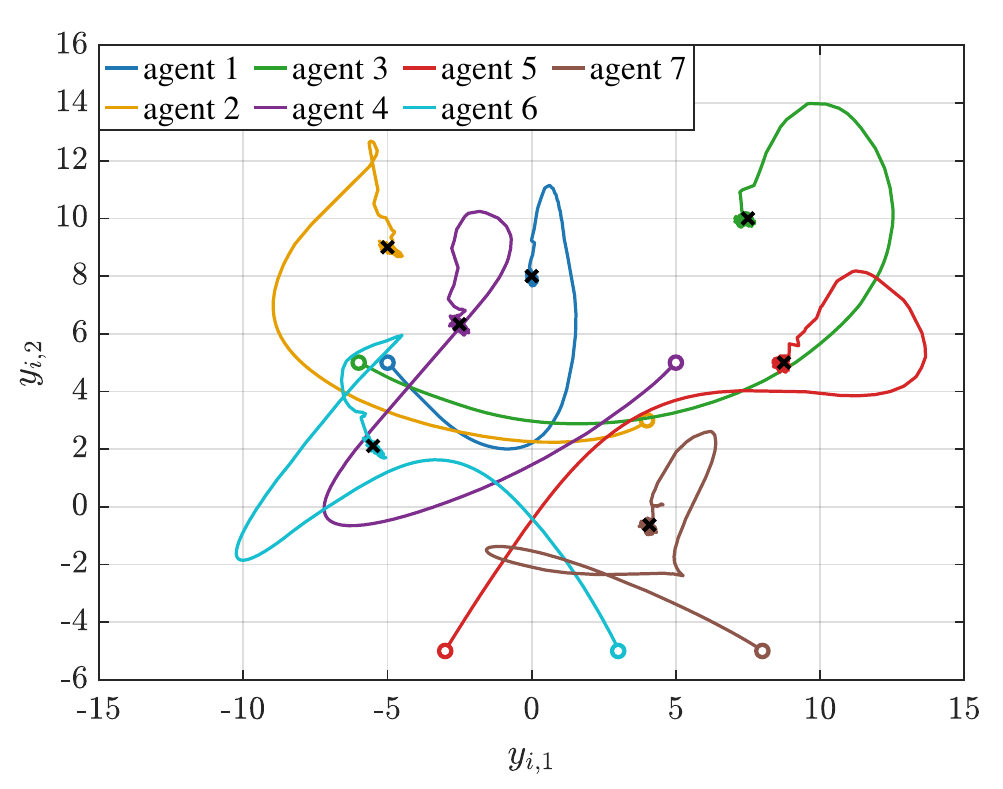}
    \caption{Decision-output trajectories under the connected
    spiking controller.}
    \label{fig:sim-connected}
\end{figure}

Figure~\ref{fig:sim-agent4-split} details the non-interacting implementation for agent 4.  The upper panels compare the impulsive velocity
components \(\upsilon_{4,\ell}\) with the corresponding auxiliary
components \(\upsilon_{4,c,\ell}\), \(\ell=1,2\). The lower panels show
the positive and negative neuronal channels and their reset impulses.
Every spike produces a jump in \(\upsilon_{4,\ell}\), accompanied by
an equal and opposite jump in the neuronal correction. The resulting
\(\upsilon_{4,c,\ell}\) remains continuous, directly illustrating the
cancellation mechanism in \eqref{eq:chi}-\eqref{eq:al}.

\begin{figure}[!htbp]
    \centering
    \includegraphics[scale=0.475]{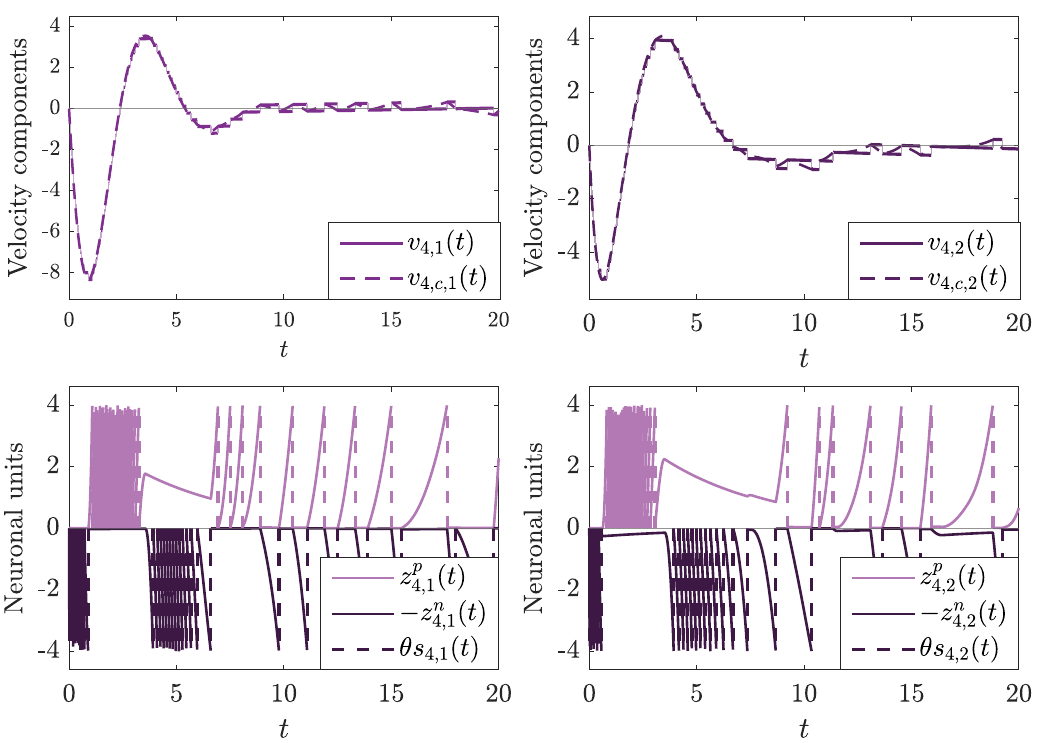}
    \caption{Spacecraft 4 under the non-interacting spiking implementation: spike-driven and auxiliary velocities (top), and neuronal states and spike-triggered resets (bottom).}
    \label{fig:sim-agent4-split}
\end{figure}

To compare the convergence rates, define the normalized output error
and pseudo-gradient norm as
\begin{equation*}
E_y(t)
\!:=\!
\frac{\|y(t)-y^\star\|}
     {\|y(0)-y^\star\|}\!\times\!100\%,~
E_F(t)
\!:=\!
\frac{\|F(y(t))\|}
     {\|F(y(0))\|}\!\times\!100\%.
\end{equation*}
Both quantities equal \(100\%\) at \(t=0\).
Figure~\ref{fig:error-comparison} shows their evolution using
logarithmic vertical scales. Under the analogue controller, both
errors converge to zero. Under the spiking implementations, the errors decrease and remain bounded, consistent with
Theorems~\ref{thm:noninteracting} and~\ref{thm:connected}. For the
selected parameters, the non-interacting implementation yields a
smaller residual than the connected implementation. This difference
depends on the impulse directions, thresholds, leakage rates, and
reset coupling used in the two implementations.

\begin{figure}[!htbp]
    \centering
    \includegraphics[scale=0.5]{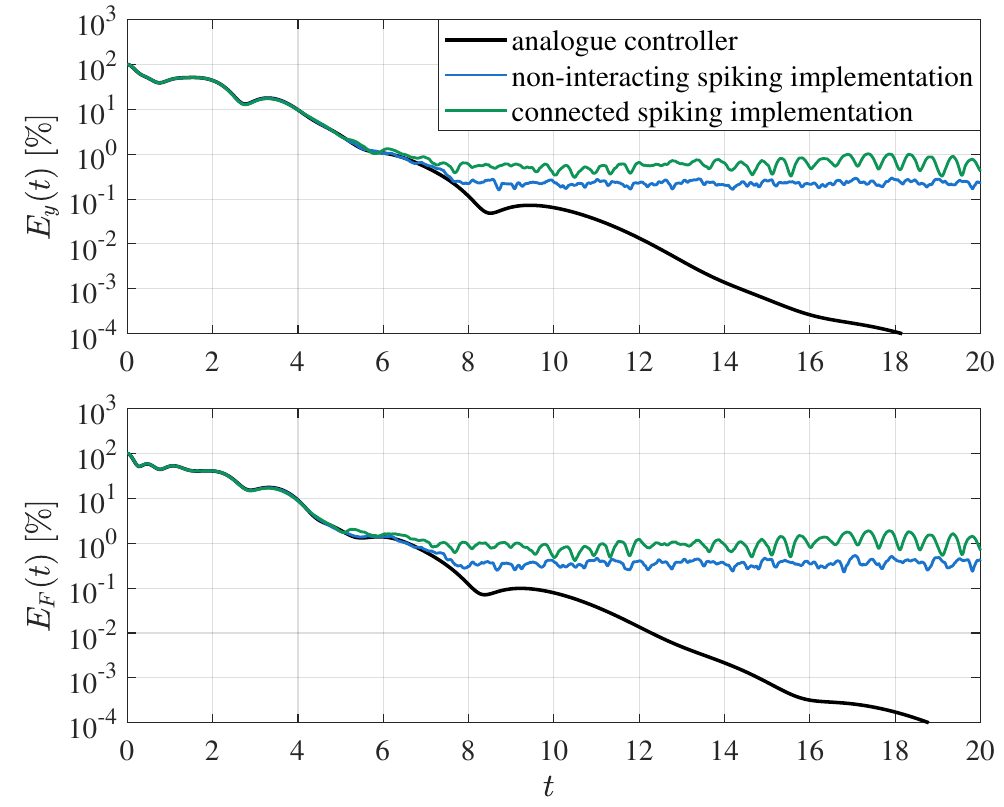}
    \caption{Output error and pseudo-gradient norm under the analogue controller and the two spiking implementations.}
    \label{fig:error-comparison}
\end{figure}

\section{Conclusions}\label{conclusion}
We developed a direct data-driven spiking control framework
for distributed $\epsilon$-NE seeking with unknown agent dynamics and
exogenous disturbances. By selecting the pseudo-gradient as the
regulated error, stabilizing analogue feedback gains are computed
from noisy local data through robust LMIs without model
identification. Online, the resulting feedback commands drive
non-interacting or connected neuronal units, while the agents receive
only fixed-weight spikes. Auxiliary coordinates that remain
continuous at event times represent the impulsive closed loops as the
corresponding stable analogue systems subject to bounded terms
depending on the neuronal variables. For non-interacting units, this
representation establishes forward completeness, excludes Zeno
behavior, and yields an ultimate bound on the pseudo-gradient. The
same properties are obtained for connected units under a spiking rule
with finite reset cascades and no finite-time event accumulation.
Consequently, both spiking implementations achieve an eventual
$\epsilon$-NE for every $\epsilon$ above a finite threshold. The
spacecraft example illustrates these theoretical results. Future work
will extend the results to general directed graphs, develop verifiable
spiking rules for connected neuronal units, and derive equilibrium
bounds directly from data.
\bibliographystyle{plainnat}
\bibliography{ref}
\end{document}